\documentclass[10pt,a4paper,reqno]{amsart}
\usepackage[hmargin={2.7cm,2.7cm},vmargin={2.5cm,2.5cm},centering]{geometry}
\usepackage{amssymb,amsmath,amsthm,amsfonts,amscd}
\usepackage{graphicx}
\usepackage[dvipsnames]{xcolor}
\usepackage{mathrsfs}
\usepackage[unicode,psdextra]{hyperref}
\usepackage[utf8]{inputenc}
\hypersetup{pdfborder={0 0 0.06},citebordercolor=[rgb]{0.8196,0.2275,0.5098},linkbordercolor=[rgb]{0.1765,0.5490,0.8118},urlbordercolor=[rgb]{0.7059,0.5333,0.1137}}
\usepackage{tikz-cd}
\usepackage[english]{babel}
\usepackage{dsfont}
\usepackage{url}
\usepackage[longnamesfirst,numbers,sort&compress]{natbib}
\usepackage{enumerate}
\usepackage{setspace}
\usepackage[font={scriptsize,bf}]{caption}
\usepackage{changepage}
\usepackage{booktabs}
\usepackage[sort&compress,capitalise,nameinlink]{cleveref}
\crefname{section}{\textsection}{\textsection}
\crefname{subsection}{\textsection}{\textsection}
\crefname{subsubsection}{\textsection}{\textsection}
\crefname{paragraph}{\textparagraph}{\textparagraph}
\crefname{thm}{Theorem}{Theorems}
\renewcommand{\Im}{\mathrm{Im}}
\renewcommand{\Re}{\mathrm{Re}}
          \newtheorem{thm}{Theorem}[section]
          \newtheorem{proposition}[thm]{Proposition}
          \newtheorem{lemma}[thm]{Lemma}
          \newtheorem{corollary}[thm]{Corollary}
          \newtheorem{definition}[thm]{Definition}
          
          \theoremstyle{definition}

          \newtheorem{remark}[thm]{Remark}

\usepackage{csquotes}
\DeclareRobustCommand{\subtitle}[1]{\\{\footnotesize #1}}

\begin{document}
\bibliographystyle{natalpha} \tolerance=2000
\setlength{\emergencystretch}{1em}

\title{Ultraviolet renormalization of the van Hove--Miyatake model\subtitle{An
    algebraic and Hamiltonian approach}}

\author[M.\ Falconi]{Marco Falconi}

\address{Politecnico di Milano\\Dipartimento di Matematica\\Campus Leonardo,
  P.zza Leonardo da Vinci 32, 20133 Milano\\Italy}

\email{marco.falconi@polimi.it} \urladdr{https://www.mfmat.org/}

\author[B.\ Hinrichs]{Benjamin Hinrichs}

\address{Universität Paderborn\\Institut für Mathematik\\Warburger Str.\
  100\\33098 Paderborn\\Germany}

\email{benjamin.hinrichs@math.upb.de}
\urladdr{https://sites.google.com/view/bhinrichs}

% \thanks{}

\dedicatory{In celebration of Hiroshima-sensei sixtieth birthday.}

% \keywords{Wigner measures, Infinite dimensional semiclassical analysis,
% Weyl algebra.}

% \subjclass[2010]{60B05, 81Q20, 81T05.}

\date{\today}

\begin{abstract}
  In this short communication we discuss the ultraviolet renormalization of
  the van Hove--Miyatake scalar field, generated by any distributional source
  $v\in \mathscr{D}'$. An abstract algebraic approach, based on the study of a
  special class of ground states of the van Hove--Miyatake dynamical map is
  compared with an Hamiltonian renormalization that makes use of a
  non-unitary dressing transformation. The two approaches are proved to yield
  equivalent results.
\end{abstract}

\maketitle
\renewcommand{\subtitle}[1]{}

\onehalfspacing{}

\section{Introduction}
\label{sec:introduction}

The van Hove--Miyatake (vHM) model is a toy model of quantum field theory,
describing the interaction of a fixed source with a bosonic quantum field and
originating in the articles \cite{vanhove1952physica,miyatake1952osaka}.
Thanks to its simplicity, it is exactly solvable and thus provides a feasible
trial platform for mathematical methods in quantum field theory: on one hand,
both its infrared and ultraviolet behavior are tractable; on the other hand,
many interesting features and problems of quantum field theory -- such as the
existence of disjoint ground states, self-energy and mass renormalization,
semiclassical analysis, scattering -- can be tested in this solvable model
\citep[see][and references
therein]{arai2020mps,derezinski2003ahp,falconi2024arxiv}.

In this note, we revisit the ultraviolet problem in the vHM model.
Considering physical massive bosons with dispersion relations
$\varpi(k)=\sqrt{\mu^2+|k|^2}$, where $k\in\mathbb R^d$ is the momentum and $\mu>0$ the
mass\footnote{The massless case $\mu=0$ plays a crucial role in discussing the
  infrared properties of the vHM model, in particular the so-called
  \emph{infrared catastrophe} \citep[see, \emph{e.g.},][and references
  therein]{arai2020mps, derezinski2003ahp}.}, the ultraviolet singularity is
reflected by the fact that the {\em source} or {\em form factor} $v:\mathbb
R^d\to\mathbb C$ fails to be square-integrable.  It is well-known that for mild
ultraviolet divergences, \emph{i.e.}, whenever $v/\varpi\in L^2(\mathbb R^d)$, the
vHM model is renormalizable by self-energy subtraction \citep[see][\textsection10.9.6
and \textsection1.1 respectively, as well as references
therein]{arai2020mps,derezinski2003ahp}. As a matter of fact, in this case,
the renormalized Hamiltonian is unitarily equivalent to the free bosonic
field in Fock representation, by means of a unitary dressing
transformation. For more singular sources this approach fails however, since
the unitarity of the dressing transformations breaks.

The purpose of this note is to study the ultraviolet problem for more
singular, even distributional, sources. We hereby compare two different
approaches:
\begin{enumerate}[(1)]
\item\label{item:4} one approach is purely algebraic, building upon an
  algebraic definition of the van Hove dynamical map implement the dynamics
  on the Weyl algebra of Canonical Commutation Relations (CCR-algebra);
\item\label{item:5} the other approach is operator theoretic, and makes use
  of the Hamiltonian formalism, by means of a \emph{non-unitary dressing
    transformation}.
\end{enumerate}
The algebraic approach \eqref{item:4} takes crucial advantage of the fact
that for this solvable (and quadratic) model, the action of the unitary Fock
dynamics -- for regular sources -- preserves the CCR-algebra, and it is
explicit on generators: it can thus be generalized without effort to singular
sources, at the abstract level of the C*-algebra and without resorting to a
representation. Furthermore, the Fock-normal ground state for regular sources
-- a coherent state centered around $-v/\varpi \in L^2$ -- can be generalized to a
non-Fock ground state whenever $-v/\varpi\notin L^2$, resorting to coherent states
centered around this singular point (that exist algebraically but are disjoint
from the Fock representation).

The operator theoretic approach \eqref{item:5} builds upon ideas of
Glimm~\citep{glimm1968cmp}, and Ginibre and
Velo~\citep{ginibre1970cmpren}. In the regular case $v/\sqrt\varpi\in L^2(\mathbb
R^d)$, we can identify the precise action of the non-unitary dressing
transformation, yielding expressions which can then be well-defined for
distributional sources, after subtraction of the self-energy of the model
(the vacuum expectation of the regularized Hamiltonian) and performing an
additional \emph{mass renormalization}: the divergent Fock vacuum expectation
of the non-unitary dressing must be used to define a new scalar product when
cutoffs are removed, thus resulting in a modified Hilbert space that carries
a representation of the CCR-algebra that is inequivalent to the Fock
representation on which the ground states for regular sources lie.  This
yields both a dressed Hilbert space, and a renormalized Hamiltonian acting on
it.

Finally, we link the two different approaches by showing that the two
renormalized Hamiltonians constructed in \eqref{item:4} and \eqref{item:5}
are \emph{unitarily equivalent}. In our opinion, this showcases the power and
limitations of either approach: algebraically, once the ground state is
obtained (and in models that are not solvable, it can be extremely difficult
to obtain), then the renormalization is automatically taken into account; on
the other hand, from the operator theoretic standpoint the renormalization
must be performed (and it often presents outstanding technical challenges),
but in doing so the ground state emerges quite naturally from the
renormalization procedure itself.

As indicated in the beginning of this introduction, treatments of the vHM
model have inspired developments for more advanced models of quantum field
theory.  In this spirit, in \cite{FalconiHinrichsValentinMartin.2025}, we
apply approach \eqref{item:5} to the spin boson model in a case where the
usual self-energy renormalization schemes -- as recently applied in
\citep[][see also references therein]{HinrichsLampartValentinMartin.2024} --
are expected to fail \cite[see][]{DamMoller.2018b}.

The rest of the paper is organized as follows. In \cref{sec:algebr-form-van}
we develop the algebraic approach \eqref{item:4}; then in
\cref{sec:hamilt-constr-van} we develop the operator theoretic approach
\eqref{item:5}, and prove unitary equivalence between the two ensuing
renormalized Hamiltonians.

\section{The algebraic formulation of the van Hove--Miyatake model}
\label{sec:algebr-form-van}

In this section we reformulate the van Hove model dynamical map on algebraic
terms. As discussed in \citep{falconi2024arxiv}, such algebraic dynamics
coincides with the standard van Hove dynamics in the Fock space whenever the
source is regular enough.

We consider the boson dispersion $\varpi\in L^\infty_{\mathrm{loc}}(\mathbb R^d)$ to be a
multiplication operator with $\varpi\ge \mu>0$ almost everywhere, generalizing the
example from the introduction.  Let $\mathscr{D} = \mathscr{D}(\mathbb{R}^d)$ be the
space of compactly supported smooth functions, and note that $\varpi$ is a
strictly positive\footnote{By strictly positive, we mean that $\langle f , \varpi f \rangle_2>
  \mu \|f\|_2$ for any $f\in \mathscr{D}$. In particular, this implies that $\varpi$ is
  also invertible on $\mathscr{D}$.} operator on $\mathscr{D}$.\footnote{For
  the massless case $\mu=0$, the space of test functions in the algebraic
  formulation is chosen differently \citep[see][]{falconi2024arxiv}, since $\lvert
  k \rvert_{}^{}$ is not a linear operator on $\mathscr{D}$.} By a slight abuse of
notation, we also denote by $\varpi: \mathscr{D}'\to \mathscr{D}'$ the operator on
$\mathscr{D}'$ obtained from $\varpi$ by transposition.  For convenience, we
identify the $L^2$-inner product $\langle \cdot , \cdot \rangle_2$ as a sesquilinear duality
bracket between $\mathscr{D}$ and $\mathscr{D}'$: $\langle \cdot , \cdot \rangle_2: \mathscr{D}\times
\mathscr{D}'\to \mathbb{C}$ by $\langle f , T \rangle_2= T(\bar{f})$, for any $f\in \mathscr{D}$ and
$T\in \mathscr{D}'$.

We see $\mathscr{D}$ as the space of test functions for a scalar quantum
field theory, and thus as customary we endow it with the natural $L^2$-inner
product $\langle \cdot , \cdot \rangle_2$, thus making $(\mathscr{D}_{\mathbb{R}},\Im \langle \cdot , \cdot \rangle_2)$ a
non-degenerate real symplectic vector space ($\mathscr{D}_{\mathbb{R}}$ is
$\mathscr{D}$ seen as a real vector space ``by doubling its basis''). Let us
denote by $\mathbb{W}(\mathscr{D}, \Im \langle \cdot , \cdot \rangle_2)$ the
C*-algebra of canonical commutation relations, generated by the Weyl
operators $\{W(f)\;,\;f\in \mathscr{D}\}$ satisfying the relations: $\forall
f,g\in \mathscr{D}$,
\begin{enumerate}[i.]
\item\label{item:2} $W(f)\neq 0$\;,
\item\label{item:1} $W(f)^{*}=W(-f)$\;,
\item\label{item:3} $W(f)W(g)= W(f+g)e^{-i\pi^2 \Im \langle f , g \rangle_2}$\;.
\end{enumerate}
The regular states on $\mathbb{W}(\mathscr{D}, \Im \langle \cdot , \cdot \rangle_2)$ are continuous positive
linear functionals $\omega$ such that for any $f\in \mathscr{D}$, $\lambda\mapsto \omega\bigl(W(\lambda
f)\bigr)$ is continuous. The {\em noncommutative Fourier transform}
\begin{equation*}
  \mathscr{D}\ni f\mapsto\hat{\omega}(f):= \omega\bigl(W(f)\bigr)\in \mathbb{C}
\end{equation*}
is a bijection \citep{segal1959mfmdvs,segal1961cjm} between regular states
and functions that are continuous when restricted to finite dimensional
subspaces of $\mathscr{D}$, and that are \emph{quantum positive definite}:
for any $\{\alpha_j\}_{j=1}^N\subset \mathbb{C}$, $\{f_j\}_{j=1}^N\subset \mathscr{D}$,
\begin{equation*}
  \sum_{j,k=1}^N\bar{\alpha}_k\alpha_j \hat\omega(f_j-f_k)e^{-i\pi^2 \Im \langle f_j  , f_k \rangle_2}\geq 0\;.
\end{equation*}
Let us denote by $\mathrm{Reg}(\mathscr{D}, \langle \,\cdot \, , \,\cdot \, \rangle_2)_+$ the set of
all regular states, and by $\mathrm{Reg}(\mathscr{D}, \langle \,\cdot \, , \,\cdot \,
\rangle_2)_{+,1}$ the set of \emph{normalized} regular states.

Finally, as before let us denote by $v\in \mathscr{D}'$ the \emph{source} of
the van Hove model.

\begin{definition}[Quantum vHM dynamical map]
  \label[definition]{absdef:18qvhmdynmap}
  For any source $v\in \mathscr{D}'$, the isometric group of *-automorphisms
  $\bigl\{\tau(t)\;,\; t\in \mathbb{R}\bigr\}$ on $\mathbb{W}(\mathscr{D}, \langle \,\cdot \, , \,\cdot \, \rangle_2)$ defined
  by extension from
  \begin{equation*}
    \tau(t)\bigl[W(f)\bigr]= W(e^{it\varpi}f)e^{2\pi i \mathrm{Re}\langle f  , (e^{-it\varpi}-1)v/\varpi \rangle_2}\;,
  \end{equation*}
  is called the \emph{quantum vHM dynamical map with source $v$}.

  \medskip
  
  \noindent Its transposed action $\tau(t)^{\mathrm{t}}$ on regular states $\omega\in
  \mathrm{Reg}(\mathscr{D}, \langle \,\cdot \, , \,\cdot \, \rangle_2)_+$ is defined by
  \begin{equation*}
    \widehat{\bigl(\tau(t)^{\mathrm{t}}[\omega]\bigr)}\bigl(f\bigr)= \hat{\omega}\bigl(e^{it\varpi}f\bigr)e^{2\pi i \mathrm{Re}\langle f  , (e^{-it\varpi}-1)v/\varpi  \rangle_{2}}\;.
  \end{equation*}
\end{definition}
Observe that such a dynamical map is defined for any source in $\mathscr{D}'$
(since $\frac{e^{-it\varpi}-1}{\varpi}$ is, by assumption, a linear map on
$\mathscr{D}'$). In particular, it is defined also whenever
$v\,,\,v/\sqrt{\varpi}\,,\,v/\varpi\notin L^2$ (it is well known that in such case the van Hove
Hamiltonian cannot be defined in the Fock representation, even taking into
account the suitable self-energy renormalization
\citep[see][]{derezinski2003ahp}).

Furthermore, we can write explicitly $(\tau,\beta)$-KMS states -- for any $\beta\leq \infty$ --
for the vHM dynamical map with source $v\in \mathscr{D}'$, and thus in
particular a \emph{ground state}.  These explicit KMS states are all regular,
and they are defined through their noncommutative Fourier transform. We leave
to the reader to check that the noncommutative Fourier transforms of the KMS
states are indeed quantum positive definite functions that are continuous
when restricted to finite dimensional subsets of $\mathscr{D}$. In order to
define such states, we proceed as below:

\begin{definition}[Algebraic Gibbs states]
  \label[definition]{absdef:18alggibbsvhm}
  For any source $v\in \mathscr{D}'$, let us define the Gibbs state at inverse
  temperature $\beta$, $\beta\leq \infty$, as the regular state $\omega_{\beta}\in
  \mathrm{Reg}(\mathscr{D}, \langle \,\cdot \, , \,\cdot \, \rangle_2)_{+,1}$ defined by the
  noncommutative Fourier transform
  \begin{equation*}
    \hat{\omega}_{\beta}(f)= e^{- \frac{\pi^2}{2}\langle f  , \mathrm{coth}(\beta\varpi/2)f \rangle_2} e^{2\pi i \mathrm{Re}\langle f  , -v/\varpi \rangle_{2}}\;.
  \end{equation*}
\end{definition}
Let us remark that $\omega_{\infty}$ is the \emph{algebraic coherent state} centered
around $-v/\varpi\in \mathscr{D}'$, whose Fourier transform is
\begin{equation*}
  \hat{\omega}_{\infty}(f)= e^{-\frac{\pi^2}{2}\langle f  , f \rangle_2}e^{2\pi i \mathrm{Re}\langle f  , -v/\varpi \rangle_2}\;,
\end{equation*}
that coincides with the usual Fock space coherent state whenever $v/\varpi \in L^2$.

\begin{proposition}
  \label[proposition]{absprop:7weakstargibbs}
  For any $\beta<\infty$, $\omega_{\beta}$ is a $(\beta,\tau)$-KMS state. Furthermore, $\omega_{\infty}$
  satisfies
  \begin{equation*}
    \omega_{\infty}= \lim_{\beta\to \infty} \omega_{\beta}\;,
  \end{equation*}
  where the limit is intended in the weak-* topology.
\end{proposition}
\begin{proof}
  A state $\omega$ is $(\tau,\beta)$-KMS by definition if and only if for any $F\in
  \mathcal{F}^{-1}\mathscr{D}$ and $a,b\in \mathbb{W}(\mathscr{D},\langle \cdot , \cdot \rangle_2)$
  \begin{equation*}
    \int_{\mathbb{R}}^{}F(t-i\beta)\omega_{\beta}\bigl(a \tau(t)[b]\bigr)  \mathrm{d}t= \int_{\mathbb{R}}^{} F(t)\omega_{\beta}\bigl(\tau(t)[b]a\bigr)  \mathrm{d}t\;.
  \end{equation*}
  Let us start by choosing $a=W(f)$ and $b=W(g)$. By an explicit computation,
  we have that
  \begin{align*}
    &
    \omega_{\beta} \bigl(W(f)\tau(t)[W(g)]\bigr)= e^{2\pi i \mathrm{Re}\langle g  , (e^{-it\varpi}-1)v/\varpi \rangle_2 - i\pi^2 \mathrm{Im}\langle f  , e^{it\varpi}g \rangle_2}\omega_{\beta} \bigl(W(f+ e^{it\varpi}g)\bigr)
    \\
    & \quad = e^{2\pi i \mathrm{Re}\langle f+g  , -v/\varpi \rangle_2}e^{- \frac{\pi^2}{2} \bigl(\langle f  , e^{it\varpi}g \rangle_2-\langle g  , e^{-it\varpi}f \rangle_2\bigr)}
    \\
    & \quad \quad \quad \times
    e^{- \frac{\pi^2}{2}\bigl(\langle f  , \mathrm{coth}(\beta\varpi/2)f \rangle_2+\langle g  , \mathrm{coth}(\beta\varpi/2)g \rangle_{2}+\langle f  , \mathrm{coth}(\beta\varpi/2)e^{it\varpi}g \rangle_2+ \langle g  , \mathrm{coth}(\beta\varpi/2)e^{-it\varpi}f \rangle_2\bigr)}\;. 
  \end{align*}
  On the other hand,
  \begin{align*}
    &\omega_{\beta} \bigl(\tau(t)[W(g)]W(f)\bigr)= e^{2\pi i \mathrm{Re}\langle f+g  , -v/\varpi \rangle_2 -\frac{\pi^2}{2} \bigl(\langle g  , e^{-it\varpi}f \rangle_2-\langle f  , e^{it\varpi}g \rangle_2\bigr)}
    \\ & \quad \quad \quad \times
    e^{-\frac{\pi^2}{2}\bigl(\langle f  , \mathrm{coth}(\beta\varpi/2)f \rangle_2+\langle g  , \mathrm{coth}(\beta\varpi/2)g \rangle_{2}+\langle f  , \mathrm{coth}(\beta\varpi/2)e^{it\varpi}g \rangle_2+ \langle g  , \mathrm{coth}(\beta\varpi/2)e^{-it\varpi}f \rangle_2\bigr)}.
  \end{align*}
  Therefore, employing the identity $(\mathrm{coth}(x/2) \pm 1)e^{\mp x} =
  \mp1 + \mathrm{coth}(x/2)$, $x\in\mathbb R$, we have
  \begin{align*}
    &\int_{\mathbb{R}}^{}F(t-i\beta) \omega_{\beta}\bigl(W(f) \tau(t)[W(g)]\bigr)
    = \int_{\mathbb{R}}^{}F(t)\omega_{\beta}\bigl(W(f) \tau(t+i\beta)[W(g)]\bigr)  \mathrm{d} t
    \\&=\int_{\mathbb{R}}^{} F(t) e^{2\pi i \mathrm{Re}\langle f+g  , -v/\varpi \rangle_2 - \frac{\pi^2}{2} \bigl(\langle f  , e^{it\varpi}e^{-\beta\varpi}g \rangle_2-\langle g  , e^{-it\varpi}e^{\beta\varpi}f \rangle_2\bigr)}
    \\&\quad\quad\quad\times
    e^{-\frac{\pi^2}{2}\bigl(\langle f  , \mathrm{coth}(\beta\varpi/2)f \rangle_2+\langle g  , \mathrm{coth}(\beta\varpi/2)g \rangle_{2}+\langle f  , \mathrm{coth}(\beta\varpi/2)e^{it\varpi}e^{-\beta\varpi}g \rangle_2+ \langle g  , \mathrm{coth}(\beta\varpi/2)e^{-it\varpi}e^{\beta\varpi}f \rangle_2\bigr)}\mathrm{d}t
    \\&= \int_{\mathbb{R}}^{}F(t)  \omega_{\beta} \bigl(\tau(t)[W(g)]W(f)\bigr)\mathrm{d}t \;.
  \end{align*}
  The result extends then to linear combinations of Weyl operators by
  linearity, and to any observable by density. The zero-temperature weak-*
  convergence of Gibbs states to $\omega_{\infty}$ is straightforward, again by first
  proving it when testing on Weyl operators and then extending by linearity
  and density.
\end{proof}

Since $\omega_{\infty}$ is the weak-* limit $\beta\to \infty$ of $(\beta,\tau)$-KMS states, it is a
\emph{ground state} or {\em $(\infty,\tau)$-KMS state}. An algebraic ground state is
a state such that for any $F\in \mathcal{F}^{-1}\mathscr{D}(\mathbb{R})$ with $\mathrm{supp}
\hat{F}\subset \mathbb{R}_{*}^-$, and any $a,b\in \mathbb{W}(\mathscr{D},\langle \cdot , \cdot \rangle_2)$:
\begin{equation*}
  \int_{\mathbb{R}}^{}F(t)\omega_{\infty}\bigl(a\tau(t)[b]\bigr)  \mathrm{d}t=0\;.
\end{equation*}
\begin{proposition}[{\citep{sirugue1971cmp}}]
  \label[proposition]{prop:1}
  The weak-* limit $\beta\to \infty$ of $(\beta,\tau)$-KMS states is a ground state.
\end{proposition}
\begin{corollary}
  \label[corollary]{cor:1}
  For any source $v\in \mathscr{D}'$, the coherent state $\omega_{\infty}$ centered
  around $-v/\varpi$ is an algebraic ground state for the vHM dynamical map.
\end{corollary}

Whenever $v/\varpi\in L^2$, the algebraic ground state $\omega_{\infty}$ is Fock normal: the
representations are unitarily equivalent, with unitary map given by
identifying the cyclic vector $\Omega_{\omega_{\infty}}$ on the GNS representation
$(\mathscr{H}_{\omega_{\infty}},\pi_{\omega_{\infty}},\Omega_{\omega_{\infty}})$ with the (cyclic) Fock coherent
state
$C(-\frac{v}{\varpi})=e^{a(\frac{v}{\varpi})-a^{*}(\frac{v}{\varpi})}\Omega_{\omega_{\mathrm{F}}}$,
that is indeed the ground state of the (renormalized/shifted) Fock vHM
Hamiltonian:
\begin{equation*}
  H_{\omega_{\mathrm{F}}} = e^{a(\frac{v}{\varpi})-a^{*}(\frac{v}{\varpi})} \hat{H}_{\omega_{\mathrm{F}}} e^{a^{*}(\frac{v}{\varpi})-a(\frac{v}{\varpi})}=e^{a(\frac{v}{\varpi})-a^{*}(\frac{v}{\varpi})}\mathrm{d}\Gamma(\varpi)e^{a^{*}(\frac{v}{\varpi})-a(\frac{v}{\varpi})}\;;
\end{equation*}
where $\hat{H}_{\omega_{\mathrm{F}}}=\mathrm{d}\Gamma(\varpi)$ is often called the
\emph{dressed renormalized vHM Hamiltonian}.

If, however, $v/\varpi\notin L^2$, then the algebraic ground state \emph{is
  disjoint from the Fock vacuum}. As a matter of fact, the GNS representation
of $\omega_{\infty}$ is a Fock representation in which the field and momentum operators
are a shift of Fock ones (the shift makes the two representations
\emph{inequivalent}) \citep[see][for a detailed
construction]{arai2020mps}. The following result can be proved rephrasing
\citep[][\textsection8.10 and \textsection10.9]{arai2020mps}.

\begin{proposition}
  \label[proposition]{prop:2}
  For any $v\in \mathscr{D}'$, the GNS representation
  $(\mathscr{H}_{\omega_{\infty}},\pi_{\omega_{\infty}},\Omega_{\omega_{\infty}})$ of the algebraic vHM ground
  state $\omega_{\infty}$ satisfies:
  \begin{itemize}
  \item $\mathscr{H}_{\omega_{\infty}}= \bigoplus_{n\in \mathbb{N}_0} L^2\bigl((\mathbb{R}^d)^n\bigr)_{\mathrm{s}}$
    is the symmetric Fock space over $L^2(\mathbb{R}^d)$\footnote{Here
      $L^2\bigl((\mathbb{R}^d)^0\bigr)_{\mathrm{s}}:= \mathbb{C}$, and for any $n>0$,
      $L^2\bigl((\mathbb{R}^d)^n\bigr)_{\mathrm{s}}$ denotes the space of square
      integrable functions that are symmetric under the permutation of
      $d$-dimensional sets of variables.}.
  \item $\Omega_{\omega_{\infty}}=(1,0,0,\dotsc,)$.
  \item For any $a\in \mathbb{W}(\mathscr{D},\langle \cdot , \cdot \rangle_2)$ and $t\in \mathbb{R}$,
    $\pi_{\omega_{\infty}}(\tau_t(a))= e^{it \mathrm{d}\Gamma(\varpi)}\pi_{\omega_{\infty}}(a)e^{-it
      \mathrm{d}\Gamma(\varpi)}$, where $\mathrm{d}\Gamma(\varpi)$ is the second quantization of
    $\varpi$.
  \end{itemize}
\end{proposition}
\begin{corollary}
  \label[corollary]{cor:2}
  For any $v\in \mathscr{D}'$, the dressed renormalized vHM Hamiltonian is
  defined, in the $\omega_{\infty}$-GNS representation, as
  \begin{equation*}
    \hat{H}_{\omega_{\infty}}= \mathrm{d}\Gamma(\varpi)\;.
  \end{equation*}
  The undressed van Hove Hamiltonian can be defined if and only if $v/\varpi\in
  L^2$, or equivalently whenever $\omega_{\infty}$ is normal with respect to the Fock
  vacuum $\omega_{\mathrm{F}}$.
\end{corollary}
This corroborates the fact that the vHM model is fundamentally
\emph{trivial}, whatever is its source $v$ (as long as it is a distribution
in $\mathscr{D}'$).

\section{A Hamiltonian construction of the van Hove--Miyatake model}
\label{sec:hamilt-constr-van}

We now move to a Hamiltonian approach to the vHM model, which we will then
prove to yield an equivalent result to \cref{prop:2,cor:2}.
It is based on a dressing transformation approach going back to Glimm
\cite{glimm1967cmpyuka1,glimm1968cmp}, and Ginibre and Velo
\cite{ginibre1970cmpren} -- from now abbreviated GGV (Glimm--Ginibre--Velo)
dressing.

Given the usual Fock space $\mathcal F = \bigoplus_{n\in \mathbb{N}_0}
L^2\bigl((\mathbb{R}^d)^n\bigr)_{\mathrm{s}}$, we define the vHM Hamiltonian as the
self-adjoint Fock space operator
\begin{align*}
  H_{\mathrm{vHM}} = \mathrm{d}\Gamma(\varpi) + a(v) + a^*(v)\;,
\end{align*}
for any source $v\in L^2(\mathbb R^d)$. This is identical to the definition
$H_{\omega_{\mathrm{F}}}$ above, but we choose slightly different notation here,
to emphasize the \emph{a priori} distinct algebraic and Hamiltonian
approaches.

Similar to the usual unitary Weyl dressing transformation discussed above,
the main idea behind the GGV dressing is to apply the -- in this case --
non-unitary $e^{a^*(-v/\varpi)}$ and rescale the wave functions with the vacuum
contribution $\|e^{a^*(-v/\varpi)}\Omega_{\mathcal F}\|^2_{\mathcal F}$, where
$\Omega_{\mathcal F}=(1,0,\ldots)$ is the usual Fock vacuum, analogously to
$\Omega_{\omega_{\mathrm{F}}}$ above. Using the canonical commutation relations (CCR)
and the Baker--Campbell--Haussdorff (BCH) formula, we can immediately prove
\begin{align*}
  \|e^{a^*(-v/\varpi)}\Omega_{\mathcal F}\|_{\mathcal F} = e^{\frac12\|v/\varpi\|_2^2}\; .
\end{align*}
At the heart of our GGV dressing then lie the following observations, which
hold on the finite particle subspace defined -- for any subspace $V\subset
L^2(\mathbb R^d)$ -- by
\[
  \mathcal F_{\mathrm{fin}} ( V ) = \operatorname{span} \big\{ f_1\otimes\cdots \otimes f_n
  \;\big| \; n\in\mathbb N \text{ and } f_1,\ldots,f_n\in V \big\}\; .
\]

\begin{proposition}
  \label[proposition]{prop:4}
  Let $\psi,\phi\in\mathcal F_{\mathrm{fin}}(\mathcal D(\varpi))$, the finite particle
  subspace over the domain of $\varpi$ seen as a multiplication operator on
  $L^2(\mathbb R^d)$, and assume $v \in L^2(\mathbb R^d)$.  Then
  \begin{align}
    \label{eq:dressedscalar}
    \frac{\langle e^{a^*(-v/\varpi)}\phi,e^{a^*(-v/\varpi)}\psi\rangle}{\|e^{a^*(-v/\varpi)}\Omega_{\mathcal F}\|_{\mathcal F}^2} 
    = \langle e^{a(-v/\varpi)}\phi, e^{a(-v/\varpi)}\psi \rangle
  \end{align}
  and
  \begin{align}
    \label{eq:dressedHamiltonian}
    \frac{\langle e^{a^*(-v/\varpi)}\phi,(H_{\mathrm{vHM}}+\|\varpi^{-1/2}v\|_2^2)e^{a^*(-v/\varpi)}\psi\rangle}{\|e^{a^*(-v/\varpi)}\Omega_{\mathcal{F}}\|_{\mathcal F}^2}
    = \langle e^{a(-v/\varpi)}\phi, \hat H_{\mathrm{vHM}} e^{a(-v/\varpi)}\psi \rangle\;,
  \end{align}
  where $\hat{H}_{\mathrm{vHM}}= \mathrm{d}\Gamma(\varpi)$ is the dressed renormalized
  van Hove Hamiltonian.
\end{proposition}
\begin{proof}
  The first statement is again immediate, by combining CCR and BCH formula.
	
  To verify the second statement, we observe the commutator identity
  \begin{align*}
    [A,e^B] = \int_0^1 e^{(1-s)B}[A,B]e^{sB}\mathrm{d}s\;,
  \end{align*}
  which holds whenever these expressions are jointly defined, as can easily
  be seen by differentiating.  Combining with the CCR, we find the operator
  identities
  \begin{align*}
    &[\mathrm{d}\Gamma(\varpi),e^{a^*(f)}] = e^{a^*(f)}a^*(\varpi f)\qquad  \text{on $\mathcal F_{\mathrm{fin}}(\mathcal D(\varpi))$},\qquad f\in \mathcal D(\varpi)\;,\\
    &[a(f),e^{a^*(g)}] =  e^{a^*(g)} \langle f,g\rangle_2\;,
    \qquad \text{on $\mathcal F_{\mathrm{fin}}(L^2(\mathbb R^d))$},\qquad f,g\in L^2(\mathbb R^d)\;,
  \end{align*}
  which again using the above argument involving CCR and BCH formula yield
  the statement.
\end{proof}
We take the right hand side of \cref{eq:dressedscalar,eq:dressedHamiltonian}
as the definition of the {\em dressed scalar product} and {\em dressed vHM
  Hamiltonian}, respectively. We will, for the remainder of this section,
argue that this construction is valid for the arbitrary distributional
sources covered in the previous section, and furthermore it comes with a
natural embedding into the Fock space, reproducing exactly the abstract GNS
representation of the above algebraic ground states.

For notational brevity, we will throughout this section adopt the notation
$g:=v/\varpi\in\mathscr D'$.  We define the corresponding annihilation operator on
$\mathcal F_{\mathrm{fin}}(\mathscr D)$ by
\begin{align*}
  a(g)(f_1\otimes_{\mathrm s} f_2\otimes_{\mathrm{s}}\cdots \otimes_{\mathrm{s}} f_n)
  = \frac{1}{\sqrt n}\sum_{\ell=1}^{n} \overline{\langle f_\ell,g \rangle_2} f_1\otimes_{\mathrm{s}}\cdots \hat f_j\cdots \otimes_{\mathrm{s}} f_n,
\end{align*}
where the $\hat{\,\cdot\,}$ denotes omission of the $j$-th factor, and extension by
linearity.  Note that $a(g)$ leaves $\mathcal F_{\mathrm{fin}}(\mathscr D)$
invariant and that $a(g)^n\psi=0$ for $\psi\in\mathcal{F}_{\mathrm{fin}}(\mathscr D)$ and $n$
large enough, whence we can define $e^{a(g)}$ as an operator on $\mathcal
F_{\mathrm{fin}}(\mathscr{D})$, by (truncated) series expansion.  Furthermore,
for $f\in\mathscr D$, we have the usual CCR
\begin{align*} [a(g),a^*(f)] = \overline{\langle f,g\rangle_2}.
\end{align*}
To verify that the right hand side of \cref{eq:dressedscalar} defines a
scalar product, the only difficulty is to check the positive definiteness.
\begin{lemma}
  \label[lemma]{lem:1}
  For any $g\in\mathscr{D}'$, the operator $e^{a(g)}$ is injective on $\mathcal
  F_{\mathrm{fin}}(\mathscr D)$.
\end{lemma}
\begin{proof}
  Assume $\psi = (\psi^{(n)})_{n\in\mathbb N_0}\in \ker e^{a(g)}$ and let $n_0\in\mathbb
  N$ such that $\psi^{(n)}=0$ for $n>n_0$. Then
  \[
    (e^{a(g)}\psi)^{(n_0)} = \psi^{(n_0)} = 0
  \]
  and thus $\psi^{(n_0)}=0$, i.e., we can replace $n_0$ by $n_0-1$. Iterating
  this argument yields $\psi = 0$ and thus injectivity of $e^{a(f)}$ ensues.
\end{proof}
This now allows us to define the scalar product
\begin{align*}
  \langle \psi, \phi \rangle_g = \langle e^{a(g)}\psi,e^{a(g)}\phi\rangle_{\mathcal F},
  \qquad \psi,\phi\in\mathcal F_{\mathrm{fin}}(\mathscr D),
\end{align*}
where the scalar product on the right is the usual Fock space scalar product.
Since this obviously equals the case $g=0$ on the left hand side, we will use
$\langle \cdot\,,\,\cdot \rangle_0$ and $\langle \cdot\,,\,\cdot \rangle_{\mathcal F}$ interchangeably from now on. We
can now thus define the dressed Hilbert space by completion of the above
inner product space.
\begin{definition}[Dressed Hilbert space]
  \label[definition]{def:1}
  Given $g\in\mathscr D'$, let $\mathcal F_g$ denote the Hilbert space
  completion of $(\mathcal F_{\mathrm{fin}}(\mathscr D),\langle \,\cdot\, , \, \cdot\, \rangle_g)$.
\end{definition}
Now, we want to use the right hand side of \cref{eq:dressedHamiltonian} to
define the dressed vHM Hamiltonian.  To this end, again for $g\in\mathscr D'$,
we define the quadratic form
\begin{align*}
  \mathfrak q_g(\phi,\psi) = \langle e^{a(g)}\phi,\mathrm{d}\Gamma(\varpi)e^{a(g)}\psi\rangle_{\mathcal F},
  \quad
  \phi,\psi\in\mathcal D(\mathfrak q_g) = \mathcal F_{\mathrm{fin}}(\mathscr D).
\end{align*}
We remark that this quadratic form is well-defined, because $\mathscr D\subset
\mathcal D(\varpi)$ and thus $\mathcal F_{\mathrm{fin}}(\mathscr D)\subset \mathcal
D(\mathrm{d}\Gamma(\varpi))$. The form is symmetric and lower-bounded, since $\varpi$ was
assumed to be strictly positive, and densely defined in $\mathcal F_g$ by
construction.

For our definition of the dressed vHM Hamiltonian, the following observation
is crucial.
\begin{lemma}
  \label[lemma]{lemma:1}
  The quadratic form $\mathfrak q_g$ is closable on $\mathcal F_g(\mathscr
  D)$ for any $g\in\mathscr D'$.
\end{lemma}
\begin{proof}
  First, recall that closability of $\mathfrak q_g$ is equivalent to the fact
  that for any sequence $(\psi_n)\subset \mathcal F_{\mathrm{fin}}(\mathscr D)$ with
  $\|\psi_n\|_g\to 0$ and $\frak q_g(\psi_n-\psi_m)\to 0$ in the usual Cauchy sense, one has $\mathfrak
  q_g(\psi_n)\to 0$.  Now since $ \mathfrak q_g(\psi) = \mathfrak q_0(e^{a(g)}\psi) $
  and $\|\psi\|_g=\|e^{a(g)}\psi\|_0$ for any $\psi\in\mathcal F_{\mathrm{fin}}(\mathscr
  D)$, this follows from closability of $\mathfrak q_0$, which in turn is
  evident from the selfadjointness of $\mathrm{d}\Gamma(\varpi)$.
\end{proof}
\begin{definition}[Dressed vHM Hamiltonian]
  \label[definition]{def:2}
  Given $g\in\mathscr D'$, let $H_g$ denote the unique selfadjoint operator on
  $\mathcal F_g$ corresponding to the closure of $\mathfrak q_g$.
\end{definition}
\begin{remark}
  \label[remark]{rem:opdomain}
  The finite particle subspace $\mathcal F_{\mathrm{fin}}(\mathscr D)$
  belongs to the {\em operator} domain of $H_g$.  This follows from the
  Cauchy--Schwarz inequality by
  \[ \frac{\mathfrak q_g(\phi,\psi)}{\|\phi\|_g} = \frac{\langle
      e^{a(g)}\phi,\mathrm{d}\Gamma(\varpi)e^{a(g)}\psi\rangle}{\|e^{a(g)}\phi\|} \le
    \|\mathrm{d}\Gamma(\varpi)e^{a(g)}\psi\|, \quad \phi,\psi\in\mathcal F_{\mathrm{fin}}(\mathscr D) \]
  and the Riesz representation theorem.  We note that this argument was used
  to construct the renormalized operator in \cite{ginibre1970cmpren}, but it
  does not yield self-adjointness.  In the proof of \cref{prop:3}, we will
  show that $\mathcal F_{\mathrm{fin}}(\mathscr D) $ is in fact a core for
  $H_g$.
\end{remark}
We can also embed the dressed Hilbert space into Fock space, and identify
exponential vectors in $\mathcal F_g$, similar to the exponential (or
coherent) vectors $e^{a^*(f)}\Omega_{\mathcal F}$, $f\in L^2(\mathbb R^d)$ in the
usual Fock space $\mathcal F$.
\begin{proposition}
  \label[proposition]{prop:dressing}
  Let $g\in\mathscr D'$. Then the following properties hold:
  \begin{enumerate}[{\em (i)}]
  \item\label{item:6} There exists a unique bounded extension $\iota_g:\mathcal F_g\to \mathcal
    F_0$ of $e^{a(g)}$.
  \item\label{item:7} $\iota_g$ is unitary.
  \item\label{item:8} The series $\epsilon_g(f) = \sum_{n=0}^\infty \frac1{\sqrt{n!}}f^{\otimes
      n}$ is absolutely convergent in $\mathcal F_g$ for any $f\in
    \mathscr{D}$. Furthermore, $\iota_g \epsilon_{g}(f) = e^{\overline{\langle f,
        g\rangle_2}}\epsilon_{0}(f)$.\footnote{Our definition of $\epsilon_0(f)$ coincides with
      the usual definition of exponential vectors in the Fock space,
      cf. \cite{Arai.2018}.}
  \end{enumerate}
\end{proposition}
\begin{proof}
  For $\psi\in\mathcal F_{\mathrm{fin}}(\mathscr{D})$, we have by definition
  $\|e^{a(g)}\psi\|_{\mathcal F}^2 = \|\psi\|_{\mathcal F_g}^2$.  This proves the
  existence and uniqueness of $\iota_g$ as well as that it is an isometric
  isometry.  To prove unitarity, it thus remain to prove that $\iota_g$ has dense
  range.  This follows from \eqref{item:8}, and the fact that the exponential
  vectors $\{\epsilon_{0}(f)\;:\; f\in \mathscr{D}\}$ span a dense subspace of $\mathcal
  F$, see for example \cite[Thm.~5.37]{Arai.2018}.
	
  It thus remains to prove \eqref{item:8}.  To this end, fix some $f\in\mathscr
  D$ and observe that by definition
  \begin{align*}
    \|f^{\otimes n}\|_{g}^2
    = \|e^{a(g)}f^{\otimes n}\|_{\mathcal F}^2
    = \sum_{\ell=0}^{n}\frac{n!}{(n-\ell)!(\ell!)^2}|\langle f,g\rangle_2|^{2\ell} \|f\|_2^{2(n-\ell)}
    \le \big(|\langle f,g\rangle|_2^2 + \|f\|^2\big)^n
    .
  \end{align*}
  This proves that the series $ \sum_{n} \frac1{{n!}}\|f^{\otimes n}\|_g^2$ is Cauchy
  and thus the claimed convergence.  By the continuity of $\iota_g$, we further
  find
  \begin{align*}
    \iota_g \epsilon_g(f)
    &= \sum_{n=0}^\infty \frac1{\sqrt{n!}}\iota_gf^{\otimes n}
    = \sum_{n=0}^\infty \frac{1}{\sqrt{n!}}e^{a(g)}f^{\otimes m} 
    = \sum_{n,m=0}^\infty\frac{1}{m!\sqrt{n!}}a(g)^mf^{\otimes m}
    \\&
    = \sum_{n,m=0}^{\infty}\frac{1}{m!\sqrt{(n-m)!}}\big(\overline{\langle f,g\rangle_2}\big)^m f^{\otimes(n-m)}
    = e^{\overline{\langle f, g\rangle_2}}\epsilon_{0}(f), \end{align*}
  where we once more used the definition of $a(g)$, as well as the absolute convergence of the double sum in the last step.
\end{proof}
Let us conclude by relating $H_g$ to the GNS representation of the algebraic
ground states.
\begin{proposition}
  \label[proposition]{prop:3}
  For $g\in\mathscr D'$, we have $ \iota_{g} H_g \iota_g^* = \mathrm{d}\Gamma(\varpi)$.
\end{proposition}
\begin{proof}
  Let $\phi,\psi\in\mathcal F_{\mathrm{fin}}(\mathscr D)$.  Then
  \begin{align*}
    \langle \phi,H_g\psi\rangle_{g}
    = \langle e^{a(g)} \phi, \mathrm{d} \Gamma (\varpi) e^{a(g)} \psi \rangle_{\mathcal F}
    = \langle \iota_{g}\phi, \mathrm{d} \Gamma (\varpi) \iota_{g} \psi \rangle_{\mathcal F},
  \end{align*}
  so by density of $\mathcal F_{\mathrm{fin}}(\mathscr D)$ we have $H_g =
  \iota_{g}^*\mathrm{d}\Gamma(\varpi)\iota_{g}$ on $\mathcal F_{\mathrm{fin}}(\mathscr D)$.  To
  prove that the operators are in fact identical, it remains to prove that
  $\iota_g\mathcal F_{\mathrm{fin}}(\mathscr D)$ is a core for $\mathrm{d}\Gamma(\varpi)$.
  Since $\iota_g$ is unitary, density hereby follows from density of $\mathcal
  F_{\mathrm{fin}}(\mathscr D) $ in $\mathcal F_g$.  Since $\iota_g\mathcal
  F_{\mathrm{fin}}(\mathscr D) \subset \mathcal F_{\mathrm{fin}}(\mathscr D) $ and
  since all $\psi\in\mathcal F_{\mathrm{fin}}(\mathscr D) $ are analytic vectors
  for $\mathrm d\Gamma(\varpi)$ w.r.t.\ the usual Fock space norm $\|\cdot\|_0$ -- by our
  assumption that $\varpi$ is locally bounded -- the claim thus follows from
  Nelson's analytic vector theorem \citep[see, \emph{e.g.}][Theorem
  X.39]{reed1975II}.
\end{proof}
It remains to study the relation of the operator theoretic ground state of
the dressed operator $\iota_{g} H_{g} \iota^{*}_g$ with the algebraic ground state
$\omega_{\infty,g}$ of \cref{absdef:18alggibbsvhm} (where we made explicit the
dependence on the source).  Clearly, the ground state of the dressed operator
in $\mathcal F_g$ is the Fock vacuum $\epsilon_g(0)$; we now prove that $\omega_{\infty,g}$ in
its GNS representation corresponds to the Fock vacuum w.r.t.\ the scalar
product in $\mathcal F_g$.  We once more recall the usual representation of
the Weyl algebra in Fock space given by
\[ \pi_0(W(f)) = e^{ i \pi(a(f)+a^*(f))}, \quad f\in\mathscr D\; .  \] Most notably, it
is uniquely characterized by the fact that
\[ \pi_0(W(f))\epsilon_{0}(h) = e^{-\frac{\pi^2}2\|f\|_2^2 + i \pi \langle f,h\rangle_2}\epsilon_{0}(h + i\pi
  f)\; . \] Furthermore, in view of \cref{prop:dressing}, we have
\[ \langle \epsilon_{g}(h) , \epsilon_{g}(f) \rangle_g = e^{\langle h,g\rangle_2 + \overline{\langle f,g\rangle_2} + \langle
    h,f\rangle_2}\;.\] Thus, the canonical choice of the Weyl representation in
$\mathcal F_g$ is uniquely given by
\[ \pi_g(W(f))\epsilon_g(h) = e^{-\frac{\pi^2}2\|f\|_2^2 + i \pi \langle f,g\rangle_2 + i \pi \langle
    f,h\rangle_2}\epsilon_{g}(h + i \pi f)\;. \] Note that the density of
$\operatorname{span}\{\epsilon_{g}(h)\,|\,h\in\mathscr D\}$ in $\mathcal F_g$, which
follows from unitarity of $\iota_g$, ensures that $\pi_g$ is well-defined.  We
leave it to the reader to verify that this definition really provides a
$*$-homomorphism from the Weyl algebra to the unitaries on $\mathcal F_g$.

Let us finally verify that $\iota_g$ maps our dressed model to the vHM ground
state, cf. \cref{absdef:18alggibbsvhm}.
\begin{proposition}
  \label[proposition]{prop:5}
  If $g\in\mathscr D'$, then $\langle \epsilon_g(0),\pi_g(W(f))\epsilon_g(0)\rangle_g = e^{-\frac 12
    \|f\|_2^2 + 2\pi i \Re \langle f,g\rangle_2} $
\end{proposition}
\begin{proof}
  This directly follows from the above observations, by
  \begin{align*}
    \langle \epsilon_{g}(0) , \pi_g(W(f))\epsilon_g(0)\rangle_g
    &	= e^{-\frac {\pi^2}2 \|f\|_2^2 + i \pi \langle f,g\rangle_2} \langle \epsilon_g(0) , \epsilon_g( i \pi f) \rangle _g 
    = e^{-\frac 12 \|f\|_2^2 + i\pi \langle f,g\rangle_2 +\overline{  \langle i \pi f,g\rangle_2}}
        \\ &
    = e^{-\frac 12 \|f\|_2^2 + 2\pi i \Re \langle f,g\rangle_2}. \qedhere
  \end{align*}
\end{proof}

\subsection*{Acknowledgements}
\label{sec:acknowledgements}

M.F.\ acknowledges the support of PNRR Italia Domani and Next Generation EU
through the ICSC National Research Centre for High Performance Computing, Big
Data and Quantum Computing; he also acknowledges the support by the Italian
Ministry of University and Research (MUR) through the grant “Dipartimento di
Eccellenza 2023-2027” of Dipartimento di Matematica, Politecnico di Milano,
and the PRIN 2022 grant ``OpeN and Effective quantum Systems (ONES)''.

{\footnotesize \bibliography{bib,00lit}
  % \printbibliography{}
}
\end{document}